\documentclass[a4paper,11pt,unpublished]{quantumarticle}
\pdfoutput=1
\usepackage[utf8]{inputenc}
\usepackage[english]{babel}
\usepackage[T1]{fontenc}
\usepackage{amsmath,amssymb,amsthm,mathtools}
\usepackage{graphicx}
\usepackage{booktabs}
\usepackage{microtype}
\usepackage{tikz}
\usetikzlibrary{positioning,arrows.meta,calc}
\usepackage[ruled,vlined,linesnumbered]{algorithm2e}
\usepackage{hyperref}

\theoremstyle{plain}
\newtheorem{theorem}{Theorem}[section]

\newtheorem{proposition}[theorem]{Proposition}

\newtheorem{conjecture}[theorem]{Conjecture}
\theoremstyle{definition}
\newtheorem{definition}[theorem]{Definition}
\newtheorem{openproblem}[theorem]{Open Problem}
\theoremstyle{remark}
\newtheorem{remark}[theorem]{Remark}

\newcommand{\Samp}{\mathsf{Samp}}
\newcommand{\Ver}{\mathsf{Ver}}
\newcommand{\negl}{\mathsf{negl}}
\newcommand{\poly}{\mathrm{poly}}
\newcommand{\GapK}{\mathsf{GapK}}
\newcommand{\Gap}{\mathsf{Gap}}
\newcommand{\pKq}{\mathrm{pK}_{\mathsf{q}}}
\newcommand{\pKqt}{\mathrm{pK}_{\mathsf{q}}^{t}}
\newcommand{\pKqs}[1]{\mathrm{pK}_{\mathsf{q}}^{#1}}
\newcommand{\pKt}{\mathrm{pK}^{t}}
\newcommand{\Kt}{\mathrm{K}^{t}}
\newcommand{\Kol}{\mathrm{K}}
\newcommand{\QPE}{\mathsf{QPE}}
\newcommand{\SD}{\mathsf{SD}}
\newcommand{\Est}{\mathsf{Estimate}}
\newcommand{\Extr}{\mathsf{Ext}}
\newcommand{\ans}{\mathsf{ans}}
\newcommand{\puzz}{\mathsf{puzz}}
\newcommand{\yes}{\mathsf{yes}}
\newcommand{\no}{\mathsf{no}}
\newcommand{\Lyes}{\mathcal{L}_{\mathsf{Yes}}}
\newcommand{\Lno}{\mathcal{L}_{\mathsf{No}}}
\newcommand{\eps}{\varepsilon}
\newcommand{\F}{\mathbb{F}}

\begin{document}

\title{Quantum Meta-Complexity Is All You Need: Characterizing One-Way Puzzles via Time-Bounded Kolmogorov Complexity}

\author{Morteza Saberikamarposhti}
\affiliation{Department of Smart Computing and Cyber Resilience, School of Computing and Artificial Intelligence, Faculty of Engineering and Technology, Sunway University, Petaling Jaya, Malaysia}

\maketitle

\begin{abstract}
We initiate the time-bounded meta-complexity program for quantum cryptography. Recent breakthroughs characterize one-way puzzles, the minimal search primitive of Microcrypt, by the average-case hardness of approximating the plain, uncomputable Kolmogorov complexity over quantumly samplable distributions; the classical program of Liu and Pass, by contrast, lives at polynomial time bounds. We define a probabilistic time-bounded quantum program complexity $\pKqt$ for classical strings and prove two unconditional theorems. First, a quantum coding theorem: any string output by a quantum polynomial-time sampler with probability $\delta$ admits a description of the information-theoretically optimal length $\log(1/\delta)$ plus logarithmic terms, decodable by a quantum machine in time $O(\sqrt{1/\delta})$ times a polynomial, via amplitude amplification over the coherently executed sampler. Second, as a consequence, an exact characterization at subexponential time: one-way puzzles exist if and only if the gap problem for $\pKqs{t^{*}}$ with $t^{*}(n) = 2^{n/2} \poly(n)$ is weakly quantum-average-hard, strictly refining the plain-complexity characterizations. We then isolate the polynomial-time coding theorem as the program's single load-bearing conjecture, prove that it implies the full polynomial-time characterization with every other component unconditional, analyze why the classical derandomization proof resists quantization, and formulate a relativized barrier conjecture: no string-valued meta-complexity problem should characterize one-way state generators, so that quantum meta-complexity on classical strings is exactly what one-way puzzles need and provably less than what state generators demand. Conjectures are labeled as such throughout; an exact numerical illustration of the probability-mass structure used by the arguments is included with reproducible code.
\end{abstract}

\section{Introduction}
\label{sec:intro}

Classical cryptography rests on the one-way function, and the one-way function now rests on meta-complexity. Liu and Pass \cite{LP20} proved that one-way functions exist if and only if the polynomial-time-bounded Kolmogorov complexity $\Kt$ is mildly hard on average; Ilango, Ren, and Santhanam \cite{IRS21} extended the equivalence to hardness over arbitrary polynomial-time samplable distributions, and the surrounding machinery of probabilistic time-bounded complexity and efficient coding theorems \cite{GKLO22,LOZ22} turned meta-complexity into a kind of assembly language for cryptography: every question about the existence of secure cryptosystems compiles down, without loss, to a question about the average-case difficulty of measuring the complexity of strings.

Quantum cryptography does not reduce to this picture, and the failure is productive. Relative to suitable oracles, quantum cryptographic primitives survive the death of one-way functions and indeed of $\mathsf{P} \neq \mathsf{NP}$ \cite{Kre21,KQST23,LMW24}. The zoo below one-way functions, Microcrypt, contains pseudorandom state generators \cite{JLS18,AQY22}, one-way state generators \cite{MY22,MY24}, one-way puzzles \cite{KT24a}, and EFI pairs \cite{BCQ23}. Among these, the meta-complexity program landed first on the \emph{one-way puzzle} of Khurana and Tomer \cite{KT24a}, and there is a conceptual gem in why. A one-way puzzle permits its verifier to be computationally unbounded: only the sampler and the adversary are efficient. Kolmogorov complexity is precisely an inefficiently verifiable quantity: the statement $\Kol(x) \le k$ has a short witness, the program, whose validity no efficient machine can certify in general. Puzzles tolerate inefficient verification, and $\Kol$ is inefficient verification incarnate; the two were made for each other. The thesis of this paper, in one sentence: \emph{one-way puzzles are equivalent to the average-case hardness of measuring the time-bounded quantum Kolmogorov complexity of quantumly samplable strings, and this equivalence is both achievable at subexponential time bounds today and, at the level of state generators, provably the end of what classical strings can express}.

The known quantum characterizations, proved concurrently by Hiroka and Morimae \cite{HM25}, by Cavalar, Goldin, Gray, and Hall \cite{CGGH25}, and, through probability estimation, by Khurana and Tomer \cite{KT24b}, establish that one-way puzzles exist if and only if a gap version of Kolmogorov complexity is weakly hard on average over quantumly samplable distributions, with further foundational study in \cite{CGGHLP25} and a state-native characterization of EFI pairs in \cite{CCCGHJL26}. Two gaps separate these results from the classical gold standard, and they are the two gaps this paper is about. First, the \emph{temporal} gap: the characterizing quantity is the plain, uncomputable $\Kol$, because the proofs compress high-probability samples by brute-force enumeration of the sampler's distribution, a description that is short and absurdly slow to decode; the classical program removed the same inefficiency with coding theorems \cite{GKLO22,LOZ22}, and no quantum coding theorem was known. Second, the \emph{descriptive} gap: the measured objects are classical strings processed by classical machines, while the native language of quantum descriptions, in the tradition of Berthiaume, van Dam, and Laplante \cite{BvDL01}, Vit\'anyi \cite{Vit01}, and G\'acs \cite{Gacs01}, has quantum machines and, ultimately, quantum states. This paper closes a precise portion of the first gap unconditionally, reduces the remainder to a single conjecture, and argues that the second gap is a one-way door: quantum \emph{machines} on classical strings are exactly what puzzles need, and classical strings of any kind are provably not what state generators need.

\section{Our results}
\label{sec:results}

\subsection*{Contributions}

In brief: we define a probabilistic time-bounded quantum program complexity $\pKqt$ for classical strings; we prove an unconditional coding theorem for quantum samplers achieving the information-theoretically optimal description length in time $O(\sqrt{1/\delta})$ (Theorem \ref{thm:codingroot}); we deduce an unconditional characterization of one-way puzzles by the average-case hardness of $\Gap\text{-}\pKqs{t^{*}}$ at $t^{*} = 2^{n/2}\poly(n)$ (Theorem \ref{thm:mainsubexp}), refining the plain-complexity characterizations of \cite{HM25,CGGH25,KT24b}; and we prove that a polynomial-time coding theorem (Conjecture \ref{conj:codingpoly}) is the single missing ingredient for the full polynomial-time characterization (Theorem \ref{thm:mainpoly}). The decoder is an instance of quantum rejection sampling \cite{ORR13}; the coding statement for quantum samplers, and the characterization it yields, are new. Assumptions: adversaries are uniform QPT throughout, hardness is over quantumly samplable distributions, and all results other than Theorem \ref{thm:mainpoly} and the conjectures are unconditional.

We state the four headline items informally; formal statements appear in the indicated sections, and each item carries its epistemic status explicitly. Theorems \ref{thm:codingroot} and \ref{thm:mainsubexp} are unconditional contributions of this paper. Theorem \ref{thm:mainpoly} is conditional on Conjecture \ref{conj:codingpoly}. Conjectures \ref{conj:codingpoly}, \ref{conj:amp}, and \ref{conj:barrier} are open, each stated with a proof strategy and its obstacles.

\paragraph{Result A: the characterization.} Fix the probabilistic time-bounded quantum program complexity $\pKqt$ of Definition \ref{def:pKq}: the least length of a classical program from which a universal quantum machine, aided by public classical randomness, reconstructs the string within $t$ steps with constant probability.

\begin{quote}
\emph{Theorem \ref{thm:mainsubexp} (unconditional, Section \ref{sec:proofA}).} One-way puzzles exist if and only if $\Gap\text{-}\pKqs{t^{*}}[\,n - n^{\eps}, n - \Delta\,]$ is weakly quantum-average-hard over some quantumly samplable distribution, for $t^{*}(n) = 2^{n/2} \poly(n)$, some $0 < \eps < 1$, and some $\Delta = \omega(\log n)$.

\emph{Theorem \ref{thm:mainpoly} (conditional on Conjecture \ref{conj:codingpoly}, Section \ref{sec:proofA}).} The same equivalence holds with $t^{*}$ replaced by a fixed polynomial.
\end{quote}

Theorem \ref{thm:mainsubexp} strictly refines the plain-complexity characterizations of \cite{HM25,CGGH25}, which correspond to the limit $t = \infty$: the characterizing quantity becomes computable in time $2^{n/2}\poly(n)$, while the adversaries remain quantum polynomial time throughout, the time bound parametrizing only the measure. The gap between $2^{n/2}$ and $\poly$ is precisely the content of the coding conjecture below.

\paragraph{Result B: the quantum coding theorem.} The engine of Result A, and the technical heart of the paper.

\begin{quote}
\emph{Theorem \ref{thm:codingroot} (unconditional, Section \ref{sec:proofB}).} If a uniform QPT sampler with running time $T(n)$ outputs $x \in \{0,1\}^n$ with probability at least $\delta$, then $\pKqt(x) \le \log(1/\delta) + O(\log(n T))$ for every $t(n) \ge \sqrt{1/\delta} \cdot (nT(n))^{O(1)}$. The description length is information-theoretically optimal up to the logarithmic term; the decoder is quantum rejection sampling \cite{ORR13}, amplitude amplification over the coherently executed sampler, filtered by an almost-universal hash whose seed rides on the public randomness.

\emph{Conjecture \ref{conj:codingpoly} (open, Section \ref{sec:proofB}).} The same bound holds with $t = \poly(n, T)$.
\end{quote}

Classically, the identical hash-and-search construction requires time $(1/\delta) \poly(n)$, so the unconditional theorem's content is a genuinely quantum quadratic speedup in decoding, and the sampler must in any case be run by a quantum machine, so no classical-machine analogue exists at any subexponential time. The polynomial-time statement quantizes the optimal coding theorem of Lu, Oliveira, and Zimand \cite{LOZ22}, whose classical proof derandomizes an Antunes-Fortnow-style argument by fixing and re-running the sampler's seed; a quantum sampler has no seed to fix, and Section \ref{sec:proofB} dissects the obstruction.

\paragraph{Result C: amplification.} The puzzle side of the amplification needed by Result A is already unconditional: one-way puzzles are existentially equivalent to their distributional variant by Chung, Goldin, and Gray \cite{CGG24}, and to non-uniform quantum pseudorandom generators \cite{KT24a,CGG24}. What remains open is amplification entirely on the meta-complexity side.

\begin{quote}
\emph{Conjecture \ref{conj:amp} (open, Section \ref{sec:ampl}).} Weak and strong quantum-average-hardness of $\Gap\text{-}\pKqt$ are equivalent, robustly in the gap width and the time bound.
\end{quote}

\paragraph{Result D: the barrier.} The theorem-shaped statement that would make the title exact.

\begin{quote}
\emph{Conjecture \ref{conj:barrier} (open, Section \ref{sec:barrier}).} There is a relativized world, built on a Haar-random unitary oracle, in which one-way state generators exist while every string-valued meta-complexity estimation problem over samplable distributions is easy. Thus quantum meta-complexity on classical strings is all one-way puzzles need, and provably less than one-way state generators demand.
\end{quote}

The unconditional evidence is assembled in Section \ref{sec:barrier}: the $\mathsf{NP}$/$\mathsf{QMA}$ barrier for puzzle characterizations of \cite{CGGH25}, the oracle worlds of \cite{Kre21,KQST23}, the one-query unitary-synthesis lower bound of \cite{LMW24}, and the fact that the EFI characterization of \cite{CCCGHJL26} was forced to state-native measures.

Figure \ref{fig:landscape} places all four results on the map.

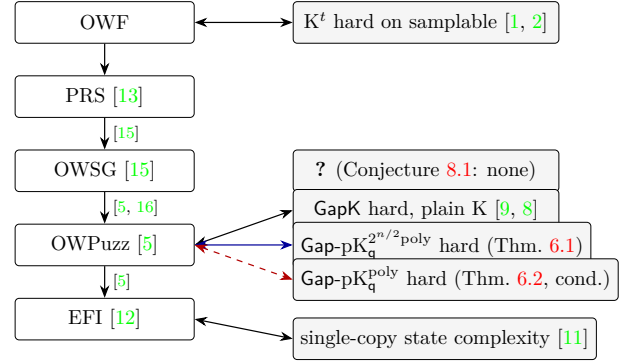
\begin{figure}[t]
\centering
\resizebox{0.98\linewidth}{!}{%
\begin{tikzpicture}[
  prim/.style={draw, rounded corners=2pt, minimum width=3.1cm, minimum height=0.72cm, align=center, font=\small},
  meta/.style={draw, rounded corners=2pt, minimum width=4.6cm, minimum height=0.72cm, align=center, font=\small, fill=black!4},
  imp/.style={-{Stealth[length=2.2mm]}, semithick},
  eqv/.style={{Stealth[length=2.2mm]}-{Stealth[length=2.2mm]}, semithick},
  neqv/.style={{Stealth[length=2.2mm]}-{Stealth[length=2.2mm]}, semithick, blue!60!black},
  ceqv/.style={{Stealth[length=2.2mm]}-{Stealth[length=2.2mm]}, semithick, dashed, red!70!black},
  lab/.style={font=\scriptsize, midway}
]
\node[prim] (owf) {OWF};
\node[prim, below=0.55cm of owf] (prs) {PRS \cite{JLS18}};
\node[prim, below=0.55cm of prs] (owsg) {OWSG \cite{MY22}};
\node[prim, below=0.55cm of owsg] (owpuzz) {OWPuzz \cite{KT24a}};
\node[prim, below=0.55cm of owpuzz] (efi) {EFI \cite{BCQ23}};

\node[meta, right=1.7cm of owf] (mowf) {$\Kt$ hard on samplable \cite{LP20,IRS21}};
\node[meta, right=1.7cm of owsg] (mowsg) {\textbf{?} (Conjecture \ref{conj:barrier}: none)};
\node[meta, right=1.7cm of owpuzz, yshift=0.60cm] (mpuzz) {$\GapK$ hard, plain $\Kol$ \cite{HM25,CGGH25}};
\node[meta, right=1.7cm of owpuzz] (mpuzzs) {$\Gap\text{-}\pKqs{2^{n/2}\poly}$ hard (Thm.~\ref{thm:mainsubexp})};
\node[meta, right=1.7cm of owpuzz, yshift=-0.60cm] (mpuzzt) {$\Gap\text{-}\pKqs{\poly}$ hard (Thm.~\ref{thm:mainpoly}, cond.)};
\node[meta, right=1.7cm of efi, yshift=-0.34cm] (mefi) {single-copy state complexity \cite{CCCGHJL26}};

\draw[imp] (owf) to (prs);
\draw[imp] (prs) to node[lab, right] {\cite{MY22}} (owsg);
\draw[imp] (owsg) to node[lab, right] {\cite{KT24a,MY24}} (owpuzz);
\draw[imp] (owpuzz) to node[lab, right] {\cite{KT24a}} (efi);

\draw[eqv] (owf) to (mowf);
\draw[eqv] (owpuzz.east) to (mpuzz.west);
\draw[neqv] (owpuzz.east) to (mpuzzs.west);
\draw[ceqv] (owpuzz.east) to (mpuzzt.west);
\draw[eqv] (efi.east) to (mefi.west);
\end{tikzpicture}}%
\caption{The meta-complexity map of Microcrypt. Left: primitives with known implications; oracle separations in the reverse directions are known or conjectured \cite{Kre21,KQST23,LMW24}. Right: characterizing problems. Black double arrows are prior equivalences; the blue double arrow is this paper's unconditional Theorem \ref{thm:mainsubexp}; the dashed red double arrow is the conditional Theorem \ref{thm:mainpoly}; the question mark is Conjecture \ref{conj:barrier}.}
\label{fig:landscape}
\end{figure}

\section{Technical overview}
\label{sec:overview}

\paragraph{Hardness implies puzzles.} Suppose $\Gap\text{-}\pKqt$ is weakly hard over a QPT sampler $\mathcal{Q}$. The route, following \cite{HM25,KT24b,CGGH25}, passes through \emph{quantum probability estimation}: the task of multiplicatively approximating the mass $p_{\mathcal{Q}}(x)$ of a given sample. If puzzles do not exist, then distributional puzzles do not exist \cite{CGG24}, and a bitwise extrapolation argument turns their absence into a QPT estimator for $p_{\mathcal{Q}}$: estimate each conditional bit probability empirically and multiply (Algorithm \ref{alg:estimate}). An estimator for the mass, in turn, decides the gap problem by thresholding (Algorithm \ref{alg:threshold}): samples of large estimated mass are declared simple, the rest complex. Soundness of the threshold has two halves. Strings declared complex but actually simple are controlled by a counting, or antichain, argument: at most $2^{k+2}$ strings have $\pKqt$ at most $k$, because a fixed program can $2/3$-produce at most one string (Proposition \ref{prop:counting}), so simple strings of small mass carry little weight. Strings declared simple must actually be simple, and this is exactly where a coding theorem is consumed: large mass under a QPT sampler must force small $\pKqt$. At time bound $2^{n/2}\poly(n)$ our unconditional Theorem \ref{thm:codingroot} supplies the implication; at polynomial time it is Conjecture \ref{conj:codingpoly}.

\paragraph{Puzzles imply hardness.} A puzzle yields, through \cite{KT24a,CGG24} and parallel repetition with padding \cite{HM25}, a non-uniform quantum pseudorandom generator whose output distribution is exponentially far from uniform yet computationally indistinguishable from it. Mix that generator's output with the uniform distribution. Uniform strings are complex except with probability $2^{-\Delta + 1}$, by counting. Generator outputs concentrate, by a mass-partition argument, on strings of mass at least $2^{-n + n^{\tau}}$ up to the advice loss, and the coding theorem compresses precisely these: the key attached to a sample is, morally, its near-minimal quantum program, whose existence the coding theorem guarantees. A decider for the gap problem therefore distinguishes the generator from uniform, contradicting indistinguishability. Both directions consume the coding theorem and nothing else beyond time-agnostic components; this is the sense, made exact in Theorem \ref{thm:mainpoly}, in which the coding theorem is the single load-bearing wall.

\paragraph{The coding theorem.} Given the sampler's code, the target's mass $\delta$, and a fingerprint of the target, decode as follows. Run the sampler \emph{coherently}, filter the output register through an almost-universal hash whose seed is read from the public randomness, and amplitude-amplify the branch whose hash matches the stored fingerprint. The fingerprint has length $\log(1/\delta)$ plus logarithms, which controls both the description length and, in expectation over the seed, the mass of colliding impostors; amplitude amplification with unknown initial amplitude, the quantum rejection sampling primitive \cite{BBHT98,BHMT02,ORR13}, reaches the matching branch in $O(\sqrt{1/\delta})$ coherent executions of the sampler; and within the matching branch the target dominates the impostors. Classically the same construction must sample until the fingerprint matches, paying $1/\delta$; the quadratic saving is irreducible for any decoder that touches the sampler only as a state-preparation oracle, by the tight query characterization of quantum resampling \cite{ORR13}, which is why we expect Conjecture \ref{conj:codingpoly} to require non-black-box use of the sampler's code, exactly as the classical proof of \cite{LOZ22} uses the sampler's tape structure.

\paragraph{The barrier, and why Haar rather than diagonalization.} For Conjecture \ref{conj:barrier} one must build a single world in which one-way state generators are secure against \emph{all} QPT adversaries with oracle access, while \emph{every} string-valued estimation problem is easy. Diagonalization constructs oracles bit by bit against a countable list of machines, which is well suited to defeating machines one at a time but poorly suited to granting a measure-one security property against all of them simultaneously; the modern route, following \cite{Kre21,KQST23}, is a Haar-random unitary oracle, whose concentration-of-measure properties give security against every efficient adversary at once, combined with a powerful classical oracle that collapses string-level complexity, the survival of unitary security against such classical power being exactly what the one-query lower bound of \cite{LMW24} evidences. Section \ref{sec:barrier} details the plan and its open steps.

An exact numerical illustration of the surprisal structure that all of the counting arguments manipulate, on a real simulated sampler, is provided in Appendix \ref{app:numerics}.

\section{Preliminaries}
\label{sec:prelim}

$n$ is the security parameter; QPT and PPT are quantum and classical probabilistic polynomial time; $\negl$ is a negligible function; $\SD$ is statistical distance. A \emph{quantumly samplable distribution} is the output distribution of a uniform QPT algorithm $\mathcal{Q}(1^n)$ over classical strings, by default $n$-bit strings. We write $p_{\mathcal{Q}}(x) := \Pr[x \leftarrow \mathcal{Q}(1^n)]$, and in subscripts abbreviate $x \leftarrow \mathcal{Q}$ and $(\ans, \puzz) \leftarrow \Samp$ with the security parameter implicit. For Kolmogorov complexity background see \cite{LV19}.

\subsection{The zoo}

\begin{definition}[One-way puzzle \cite{KT24a}]
\label{def:owpuzz}
A one-way puzzle is a pair $(\Samp, \Ver)$: $\Samp(1^n) \rightarrow (\ans, \puzz)$ is QPT and outputs two classical strings; $\Ver(\ans', \puzz) \rightarrow \top/\bot$ is unbounded. Correctness:
\begin{equation}
\Pr_{(\ans,\puzz) \leftarrow \Samp}\big[\top \leftarrow \Ver(\ans, \puzz)\big] \ge 1 - \negl(n).
\end{equation}
Security: for every uniform QPT adversary $\mathcal{A}$ (receiving $1^n$ implicitly),
\begin{equation}
\Pr_{(\ans,\puzz) \leftarrow \Samp}[\top \leftarrow \Ver(\mathcal{A}(\puzz), \puzz)] \le \negl(n).
\end{equation}
\end{definition}

One-way state generators (OWSG) \cite{MY22,MY24} replace the classical puzzle by copies of a quantum state; EFI pairs \cite{BCQ23} are efficiently samplable, statistically far, computationally indistinguishable state pairs, equivalent to quantum bit commitments \cite{BCQ23,Yan22} and implied by puzzles \cite{KT24a}, with secure computation following \cite{BCKM21,GLSV21}. The implication diagram is in Figure \ref{fig:landscape}.

\subsection{Complexity measures}

Fix a universal Turing machine $U$ and a universal quantum machine $U_{\mathsf{Q}}$ that takes a classical program $d$ and a classical auxiliary string $r$, runs for an allotted number of steps, and measures a designated output register; $U_{t}$ and $U_{\mathsf{Q},t}$ denote the machines halted after $t(|x|)$ steps. Then $\Kol(x) := \min\{|d| : U(d) = x\}$ and $\Kt(x) := \min\{|d| : U_{t}(d) = x\}$. The probabilistic time-bounded complexity $\pKt(x)$ of \cite{GKLO22} is the least $k$ such that
\begin{equation}
\Pr_{r}\big[\exists d,\ |d| \le k :\ U_{t}(d, r) = x\big] \ge \tfrac{2}{3},
\end{equation}
with $r \leftarrow \{0,1\}^{t(|x|)}$ uniform: public randomness available to the existential quantifier.

\begin{definition}[Probabilistic time-bounded quantum program complexity]
\label{def:pKq}
$\pKqt(x)$ is the least $k$ such that
\begin{equation}
\Pr_{r}\big[\, \exists d,\ |d| \le k :\ \Pr[U_{\mathsf{Q},t}(d, r) = x] \ge \tfrac{2}{3} \,\big] \ge \tfrac{2}{3},
\end{equation}
with $r \leftarrow \{0,1\}^{t(|x|)}$ uniform and the inner probability over the machine's measurements.
\end{definition}

This is the direct quantum analogue of $\pKt$: public classical randomness, classical program, quantum decoder. Classical programs are forced on us by the goal, a \emph{string}-input meta-complexity problem; qubit-program measures in the tradition of \cite{BvDL01,Gacs01} and classical-description state measures \cite{Vit01} belong to the state-native axis of \cite{CCCGHJL26} and reappear in Section \ref{sec:barrier}.

\begin{proposition}[Counting]
\label{prop:counting}
For all $t$, $n$, $k$: $\ \big|\{x \in \{0,1\}^n : \pKqt(x) \le k\}\big| \le 2^{k+2}$.
\end{proposition}

\begin{proof}
Fix $r$ and let $S_r := \{x : \exists d, |d| \le k, \Pr[U_{\mathsf{Q},t}(d,r) = x] \ge 2/3\}$. Two distinct strings cannot both be output with probability at least $2/3$ by the same $(d, r)$, so $|S_r| \le 2^{k+1} - 1$. If $\pKqt(x) \le k$ then $\Pr_r[x \in S_r] \ge 2/3$, hence $\tfrac{2}{3}\,\big|\{x : \pKqt(x) \le k\}\big| \le \mathbb{E}_r\big[|S_r|\big] \le 2^{k+1}$, giving the bound $3 \cdot 2^{k} \le 2^{k+2}$.
\end{proof}

\begin{proposition}[Classical simulation]
\label{prop:simulation}
There are a polynomial $p$ and a constant $c$ with $\pKqs{p(t)}(x) \le \pKt(x) + c$ and $\pKqs{p(t)}(x) \le \Kt(x) + c$ for all $x, t$.
\end{proposition}

\begin{proof}
$U_{\mathsf{Q}}$ simulates $U$ with polynomial overhead and success probability $1$; prepend a constant-size selector.
\end{proof}

\begin{proposition}[Amplification]
\label{prop:amplification}
For every $\lambda$, the inner threshold $2/3$ in Definition \ref{def:pKq} may be replaced by $1 - 2^{-\lambda}$ at multiplicative cost $O(\lambda)$ in time and additive $O(1)$ in program length: the modified universal machine reruns the program $O(\lambda)$ times and outputs the majority string, which is correct by a Chernoff bound because the target is the unique outcome of per-run probability above $1/2$. Consequently all constants in $(1/2, 1)$ define the same measure up to additive constants, and similarly for the outer threshold by standard arguments.
\end{proposition}

\subsection{Gap problems, hardness, probability estimation}

For a complexity measure $\mu$ and $s_2 - s_1 = \omega(\log n)$, the promise problem $\Gap\mu[s_1, s_2] = (\Lyes, \Lno)$ has yes-instances $\mu(x) \le s_1(|x|)$ and no-instances $\mu(x) \ge s_2(|x|)$; for $\mu = \Kol$ this is the $\GapK$ convention of \cite{HM25,CGGH25}, itself in the GapMINKT tradition. Weak quantum-average-hardness, following \cite{HM25}: there are $k > 1$ and a QPT sampler $\mathcal{Q}$ such that every uniform QPT $\mathcal{A}$ satisfies, for all large $n$,
\begin{multline}
\label{eq:weakhard}
\Pr_{x \leftarrow \mathcal{Q}}\big[\no \leftarrow \mathcal{A}(x) \wedge x \in \Lyes\big] \\ + \Pr_{x \leftarrow \mathcal{Q}}\big[\yes \leftarrow \mathcal{A}(x) \wedge x \in \Lno\big] \ge n^{-k}.
\end{multline}
Quantum probability estimation ($\QPE$) is quantum-average-hard \cite{HM25,KT24b} if there are $c > 1$, $q$, and a QPT sampler $\mathcal{Q}$ such that every QPT $\Est$ satisfies, for all large $n$,
\begin{equation}
\label{eq:qpe}
\Pr_{x \leftarrow \mathcal{Q}}\big[\tfrac{1}{c} p_{\mathcal{Q}}(x) \le \Est(x) \le c\, p_{\mathcal{Q}}(x)\big] \le 1 - n^{-q}.
\end{equation}
The known characterization \cite{HM25,CGGH25,KT24b}: one-way puzzles exist iff $\GapK[n - n^{\eps}, n - \Delta]$ is weakly quantum-average-hard for some $\eps, \Delta$, iff $\QPE$ is quantum-average-hard. We use two of its components as black boxes: the distributional-puzzle equivalence of \cite{CGG24}, and the equivalence of puzzles with non-uniform quantum pseudorandom generators admitting exponentially-far amplification \cite{KT24a,CGG24,HM25}.
\section{The quantum coding theorem}
\label{sec:proofB}

\begin{theorem}[Optimal-length coding in root time]
\label{thm:codingroot}
There is a constant $c$ such that the following holds. Let $\mathcal{Q}$ be a uniform QPT sampler with running time $T(n) \ge n$, let $n$ be sufficiently large, let $\delta \in (0,1]$, and let $t$ be any time bound with $t(n) \ge \sqrt{1/\delta}\,(n T(n))^{c}$. Then every $x \in \{0,1\}^n$ with $p_{\mathcal{Q}}(x) \ge \delta$ satisfies
\begin{equation}
\pKqt(x) \ \le\ \log(1/\delta) + c \log\big(n T(n)\big).
\end{equation}
\end{theorem}

\begin{proof}
Write $p := p_{\mathcal{Q}}(x) \ge \delta$ and set $m := \lceil \log(1/\delta) \rceil + 4 \lceil \log(n T) \rceil$.

\emph{The hash family.} Identify $\{0,1\}^m$ with the field $\F_{2^m}$ and split $z \in \{0,1\}^n$ into $\ell := \lceil n/m \rceil$ blocks $z_1, \ldots, z_\ell \in \F_{2^m}$, padding with zeros. For a seed $a \in \F_{2^m}$ define $h_a(z) := \sum_{i=1}^{\ell} z_i\, a^{i-1} \in \F_{2^m}$. For distinct $z, z'$ the difference polynomial is nonzero of degree less than $\ell$, so
\begin{equation}
\label{eq:collision}
\Pr_{a \leftarrow \F_{2^m}}\big[h_a(z) = h_a(z')\big] \le \frac{\ell - 1}{2^m} < \frac{n}{2^m};
\end{equation}
this is the polynomial-evaluation instance of universal hashing \cite{CW79}, computable in time $\poly(n)$.

\emph{The description.} Let the public random string $r$ supply the seed $a$ (its first $m$ bits). The program is $d := (\textsf{header}, n, v)$ where $\textsf{header}$ contains the constant-size code of the decoder below together with the constant-size code of the uniform sampler $\mathcal{Q}$, and $v := h_a(x)$. Then $|d| = m + O(\log n) \le \log(1/\delta) + c \log(nT)$ for suitable $c$.

\emph{Collision mass.} Let $B_a := \sum_{y \neq x,\ h_a(y) = h_a(x)} p_{\mathcal{Q}}(y)$. By \eqref{eq:collision} and linearity, $\mathbb{E}_a[B_a] \le n\, 2^{-m} \le \delta\, n\, (nT)^{-4} \le \delta\, (nT)^{-3}$. By Markov's inequality,
\begin{equation}
\Pr_{a}\big[B_a > \delta/20\big] \le 20\,(nT)^{-3} \le \tfrac{1}{10}
\end{equation}
for large $n$. Call $a$ \emph{good} otherwise; goodness depends only on $r$.

\emph{The decoder.} On $(d, r)$: parse $a$ and $v$; let $A$ be the unitary that runs $\mathcal{Q}(1^n)$ coherently, all measurements deferred, so that $A|0\rangle = \sum_z \sqrt{p_{\mathcal{Q}}(z)}\, |z\rangle |\varphi_z\rangle$ for some garbage states $|\varphi_z\rangle$; let $\Pi$ project onto basis states of the output register with $h_a(z) = v$, implemented coherently in $\poly(n)$ gates by computing the hash, comparing, and uncomputing. Run the amplitude amplification search with unknown initial success probability, an instance of quantum rejection sampling \cite{BBHT98,BHMT02,ORR13}, on $(A, \Pi)$, truncated after $C \sqrt{1/\delta}$ rounds for an absolute constant $C$, and output the measured output register.

\emph{Analysis.} The initial success probability is $s_a := p + B_a \ge \delta$, so the truncated search returns a sample from the marked subspace with probability at least $9/10$ \cite{BBHT98}. Amplitude amplification rotates within the plane spanned by $\Pi A |0\rangle$ and its complement, preserving relative amplitudes inside the marked subspace; hence, conditioned on success, the output is distributed as $p_{\mathcal{Q}}(\cdot)/s_a$ on hash-matching strings, and equals $x$ with probability
\begin{equation}
\frac{p}{p + B_a} \ \ge\ \frac{p}{p + \delta/20}\ \ge\ \frac{20}{21}
\end{equation}
whenever $a$ is good, using $p \ge \delta$. For good $a$ the inner success probability is therefore at least $(9/10)(20/21) > 2/3$, and $a$ is good with probability at least $9/10 > 2/3$ over $r$. The decoder's running time is $O(\sqrt{1/\delta})$ coherent executions of $A$, $A^{-1}$, and the $\poly(n)$-gate reflection, hence at most $\sqrt{1/\delta}\,(nT)^{c}$ for suitable $c$, within the allotted $t$. All requirements of Definition \ref{def:pKq} are met.
\end{proof}

\begin{remark}[What is quantum here]
\label{rem:whatisquantum}
Two things, and both are essential. First, decoding a quantum sampler's output requires running the sampler, so unless QPT sampling is classically simulable, no classical-machine measure admits a subexponential coding theorem for quantum samplers; the inefficient enumeration of \cite{HM25,CGGH25} certifies plain $\Kol$ only. Second, the classical version of the identical construction, sample and compare fingerprints, needs $\Theta(1/\delta)$ trials, and indeed $\Omega(1/\delta)$ black-box probes are necessary even to estimate the mass classically \cite{LOZ22}; the quadratic saving is exactly Grover's. For \emph{classical} samplers, by contrast, seed-space techniques achieve optimal length in polynomial time \cite{LOZ22}, so the quantum case is where root time is currently the frontier. In the regime the characterization needs, $\delta \approx 2^{-n + n^{\tau}}$, the improvement is from $2^{n - o(n)}$ to $2^{n/2 - o(n)}$ time. A numerical validation of the full construction, exercising the hash, the amplification schedule, and the conditional-distribution step on exactly simulated samplers, is reported in Appendix \ref{app:decoder}.
\end{remark}

\begin{conjecture}[Polynomial-time quantum coding]
\label{conj:codingpoly}
Theorem \ref{thm:codingroot} holds with $t(n) = (n T(n))^{O(1)}$, independent of $\delta$.
\end{conjecture}

\begin{remark}[Obstacles and strategy]
\label{rem:obstacles}
The optimal classical theorem, for $\pKt$, is proved in \cite{LOZ22} through seed-space hashing: the public randomness specifies a function $H_w$ mapping a short index $v$, of length $\log(1/\delta) + O(\log T)$, into the sampler's random tape, generated by a pseudorandom generator against constant-depth circuits, such that for most $w$ some index satisfies $M(H_w(v)) = x$; decoding is then a single deterministic run of the sampler on the tape $H_w(v)$, in polynomial time. Every ingredient manipulates the deterministic map from random tape to output. A quantum sampler has no such map: its randomness is amplitude, there is no tape to index, and the certifying predicate, that some index drives the sampler to $x$, is not even well defined. The black-box route is closed sharply: the query complexity of quantum resampling is tightly characterized \cite{ORR13}, so any decoder with only state-preparation oracle access to the sampler inherits the $\Omega(\sqrt{1/\delta})$ behavior of Theorem \ref{thm:codingroot}, and a proof of the conjecture must therefore use the \emph{code} of the sampler non-trivially, exactly as \cite{LOZ22} uses the tape structure. We also note that the conditional lower bounds of \cite{LOZ22} against efficient coding under cryptographic exponential-time hypotheses target Levin-style measures with efficient compressors, and do not apply to the public-coin, existential $\pKq$ regime; no barrier of that type is known to block the conjecture. Candidate ingredients: quantum rewinding, gentle-measurement or shadow-tomographic access to intermediate states of $A$, and hybrid descriptions carrying classical transcripts of measurements interleaved with coherent segments. A refutation would be equally significant: a QPT-samplable family of high-mass, incompressible-in-quantum-polynomial-time strings would be a natural hardness candidate, though we caution that hardness of estimation and hardness of compression are not known to imply one another.
\end{remark}

\section{The characterization: proofs of the main theorems}
\label{sec:proofA}

\begin{theorem}[Unconditional characterization at subexponential time]
\label{thm:mainsubexp}
Let $t^{*}(n) := 2^{n/2} n^{c_0}$ for a suitable absolute constant $c_0$. The following are equivalent.
\begin{enumerate}
\item[(i)] One-way puzzles exist.
\item[(ii)] There exist $0 < \eps < 1$ and polynomial-time-computable $\Delta(n) = \omega(\log n)$ such that $\Gap\text{-}\pKqs{t^{*}}[\,n - n^{\eps},\, n - \Delta\,]$ is weakly quantum-average-hard.
\item[(iii)] $\QPE$ is quantum-average-hard.
\end{enumerate}
\end{theorem}

\begin{theorem}[Conditional characterization at polynomial time]
\label{thm:mainpoly}
Assume Conjecture \ref{conj:codingpoly}. Then Theorem \ref{thm:mainsubexp} holds with $t^{*}$ replaced by a fixed polynomial.
\end{theorem}

\begin{proof}[Proof of Theorems \ref{thm:mainsubexp} and \ref{thm:mainpoly}]
We prove the cycle (iii) $\Rightarrow$ (i) $\Rightarrow$ (ii) $\Rightarrow$ (iii). The first edge is unconditional and due to prior work; the other two consume a coding theorem exactly once each, instantiated by Theorem \ref{thm:codingroot} at $t = t^{*}$ for Theorem \ref{thm:mainsubexp} and by Conjecture \ref{conj:codingpoly} for Theorem \ref{thm:mainpoly}; all masses to be compressed below are at least $2^{-n}$, for which $t^{*}(n) = 2^{n/2} n^{c_0}$ suffices in Theorem \ref{thm:codingroot}. Parameter bookkeeping follows \cite{HM25} and we present the argument at the level of that source's lemmas, flagging our replacements.

\emph{(iii) $\Rightarrow$ (i).} This is the unconditional direction of \cite{HM25,KT24b,CGGH25}. In contrapositive: if puzzles do not exist then distributional puzzles do not exist \cite{CGG24}, yielding for every QPT sampler a bitwise extrapolator $\Extr$ whose conditional distributions are inverse-polynomially correct; Algorithm \ref{alg:estimate} multiplies empirical estimates of the $n$ conditional bit probabilities into a multiplicative estimator for $p_{\mathcal{Q}}$, whose accumulated-error analysis is Lemma 4.3 of \cite{HM25}. Hence $\QPE$ is easy.

\emph{(ii) $\Rightarrow$ (iii).} In contrapositive, following the shape of Lemma 4.2 of \cite{HM25}, itself sharpening \cite{IRS21}. Let $\mathcal{Q}$ witness (ii) with thresholds $s_1 = s - \Delta$, $s_2 = s$, $s \le n$, and suppose a QPT $\Est$ achieves \eqref{eq:qpe} with constant $c' \in (1, 100/99]$. Algorithm \ref{alg:threshold} outputs $\yes$ iff $\Est(y) \ge 2^{-s + \Delta/2}$. Error events: (a) $y \in \Lyes$, that is $\pKqs{t}(y) \le s - \Delta$, but $\Est(y)$ small. When the estimator is accurate this forces $p_{\mathcal{Q}}(y) < \tfrac{100}{99} 2^{-s+\Delta/2}$; by Proposition \ref{prop:counting} there are at most $2^{s - \Delta + 2}$ such simple strings, so their total mass is at most $2^{s-\Delta+2} \cdot \tfrac{100}{99} 2^{-s + \Delta/2} \le 2^{-\Delta/3}$ for large $n$: this is Claim 4.6 of \cite{HM25} with plain-$\Kol$ counting replaced by Proposition \ref{prop:counting}, at the cost of an immaterial constant. (b) $y \in \Lno$, that is $\pKqs{t}(y) \ge s$, but $\Est(y)$ large. Accuracy forces $p_{\mathcal{Q}}(y) \ge \tfrac{99}{100} 2^{-s + \Delta/2} \ge 2^{-n}$, and the coding theorem (at the appropriate $t$) gives $\pKqs{t}(y) \le s - \Delta/2 + O(\log n) < s$, a contradiction; so event (b) has probability zero up to estimator failure. This replaces the inefficient enumeration machine of Claim 4.7 of \cite{HM25}, the single point where plain $\Kol$ was previously required. Summing, the adversary violates \eqref{eq:weakhard}, contradicting (ii).

\emph{(i) $\Rightarrow$ (ii).} Puzzles yield non-uniform quantum pseudorandom generators \cite{KT24a,CGG24}, amplifiable by parallel repetition and padding to be $(1 - 2^{-n^{\tau}})$-far from uniform while computationally indistinguishable, with $\log n$ bits of advice (Lemmas 2.4 and 2.7 of \cite{HM25}). Fix $\eps$, set $\tau := (1+\eps)/2$ and $G := 2^{-n^{\eps}}$, and let $\mathcal{Q}$ sample advice $\mu \leftarrow [n]$ and output, with probability $1/2$ each, a uniform string or a generator sample $\mathsf{Gen}(1^n, \mu)$. The mass-partition argument of Claim 4.12 of \cite{HM25} shows that except with probability $G + 2^{-n^{\tau}}$, a sample of $\mathsf{Gen}(1^n, \mu^{*})$ under the good advice lands in the set $C$ of strings with $p_{\mathcal{Q}} \ge \tfrac{1}{2n}\, G\, 2^{-n + n^{\tau}} \ge 2^{-n}$; the coding theorem compresses every $y \in C$ to
\begin{equation}
\pKqs{t}(y) \ \le\ n - n^{\tau} + n^{\eps} + O(\log n) \ \le\ n - n^{\eps}
\end{equation}
for large $n$, so $C \subseteq \Lyes$: the second and final replacement of the enumeration machine. Uniform samples satisfy $\pKqs{t} \ge n - \Delta$ except with probability at most $2^{-\Delta + 2}$ by Proposition \ref{prop:counting}, so they are no-instances up to negligible loss. A QPT decider beating the trivial error on $\mathcal{Q}$ by $n^{-k}$ therefore distinguishes $\mathsf{Gen}(1^n, \mu^{*})$ from uniform with advantage $n^{-k-1}$, the $1/n$ advice loss absorbed as in Theorem 4.5 and Lemma 4.4 of \cite{HM25}, contradicting indistinguishability and hence (i)'s security. This establishes (ii) and closes the cycle.
\end{proof}

\begin{remark}
The two coding invocations are the only time-sensitive steps: counting, extrapolation, distributional equivalence, and generator amplification are agnostic to $t$. Theorem \ref{thm:mainsubexp} is therefore the best unconditional point currently reachable on the time axis, and Conjecture \ref{conj:codingpoly} is equivalent to sliding it to polynomial time by exactly this proof.
\end{remark}

\begin{algorithm}[t]
\caption{$\Est(y)$: probability estimation from bitwise extrapolation \cite{CGG24,HM25}}
\label{alg:estimate}
\KwIn{$y = (y_1, \ldots, y_n) \in \{0,1\}^n$; oracle access to the QPT extrapolator $\Extr$; accuracy parameter $q$}
\KwOut{an estimate of $p_{\mathcal{Q}}(y)$}
\For{$i = 1$ \KwTo $n$}{
  run $b \leftarrow \Extr(1^n, i, y_1, \ldots, y_{i-1})$ independently $N = n^{100 + 100q}$ times\;
  $\widetilde{p}[y_i] \leftarrow \#\{\text{runs outputting } y_i\} / N$\;
}
\Return $\prod_{i=1}^{n} \widetilde{p}[y_i]$\;
\end{algorithm}

\begin{algorithm}[t]
\caption{Threshold decision for $\Gap\text{-}\pKqt[s - \Delta, s]$ from a $\QPE$ estimator \cite{HM25,IRS21}}
\label{alg:threshold}
\KwIn{$y \leftarrow \mathcal{Q}(1^n)$; a QPT estimator $\Est$; thresholds $s, \Delta$}
\KwOut{$\yes$ or $\no$}
compute $e \leftarrow \Est(y)$\;
\eIf{$e \ge 2^{-s+\Delta/2}$}{\Return $\yes$ \tcp*[r]{sound by the coding theorem}}{\Return $\no$ \tcp*[r]{sound by Proposition \ref{prop:counting}}}
\end{algorithm}

\begin{algorithm}[t]
\caption{The decoder of Theorem \ref{thm:codingroot}}
\label{alg:decoder}
\KwIn{program $d = (\textsf{header}, n, v)$; public randomness $r$ supplying the seed $a$}
\KwOut{a string in $\{0,1\}^n$}
build the coherent sampler $A$ from the code of $\mathcal{Q}$ in \textsf{header}\;
build the reflection about $\{z : h_a(z) = v\}$ via hash, compare, uncompute\;
run amplitude amplification with unknown initial amplitude \cite{BBHT98,BHMT02}, truncated at $C\sqrt{1/\delta}$ rounds\;
\Return the measured output register\;
\end{algorithm}

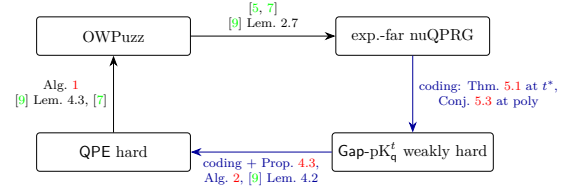
\begin{figure}[t]
\centering
\resizebox{0.9\linewidth}{!}{%
\begin{tikzpicture}[
  box/.style={draw, rounded corners=2pt, minimum width=3.4cm, minimum height=0.85cm, align=center, font=\small},
  uncond/.style={-{Stealth[length=2.4mm]}, semithick},
  cond/.style={-{Stealth[length=2.4mm]}, semithick, blue!60!black},
  lab/.style={font=\scriptsize, align=center}
]
\node[box] (puzz) {OWPuzz};
\node[box, right=3.2cm of puzz] (gen) {exp.-far nuQPRG};
\node[box, below=1.7cm of gen] (gap) {$\Gap\text{-}\pKqt$ weakly hard};
\node[box, below=1.7cm of puzz] (qpe) {$\QPE$ hard};

\draw[uncond] (puzz) to node[lab, above] {\cite{KT24a,CGG24}\\ \cite{HM25} Lem.~2.7} (gen);
\draw[cond] (gen) to node[lab, right] {coding: Thm.~\ref{thm:codingroot} at $t^{*}$,\\ Conj.~\ref{conj:codingpoly} at $\poly$} (gap);
\draw[cond] (gap) to node[lab, below] {coding $+$ Prop.~\ref{prop:counting},\\ Alg.~\ref{alg:threshold}, \cite{HM25} Lem.~4.2} (qpe);
\draw[uncond] (qpe) to node[lab, left] {Alg.~\ref{alg:estimate}\\ \cite{HM25} Lem.~4.3, \cite{CGG24}} (puzz);
\end{tikzpicture}}%
\caption{Architecture of the characterization. Black edges are unconditional prior work; blue edges consume a coding theorem, unconditionally supplied at $t^{*} = 2^{n/2}\poly(n)$ by Theorem \ref{thm:codingroot} and at polynomial time by Conjecture \ref{conj:codingpoly}.}
\label{fig:architecture}
\end{figure}

\section{Amplification and robustness}
\label{sec:ampl}

On the puzzle side, amplification is settled: one-way puzzles are existentially equivalent to distributional one-way puzzles \cite{CGG24} and to non-uniform quantum pseudorandom generators \cite{KT24a,CGG24}, and these equivalences are what the proofs above consume. On the measure side, Proposition \ref{prop:amplification} makes the success thresholds of $\pKqt$ inessential. What remains open is amplification of the \emph{hardness} itself.

\begin{conjecture}[Meta-side amplification]
\label{conj:amp}
Weak quantum-average-hardness of $\Gap\text{-}\pKqt$ for some gap $\Delta = \omega(\log n)$ is equivalent to its strong quantum-average-hardness, with error $1/2 - 1/\poly$, for related parameters; and the hardness is robust under widening the gap to $n^{\eps'}$ and under all sufficiently large polynomial time bounds simultaneously.
\end{conjecture}

Classically the corresponding translations are supplied by \cite{IRS21,GKLO22}, powered by symmetry-of-information and language-compression theorems for $\pKt$; the $\pKqt$ analogues of those structural tools are open, and direct-product amplification must additionally survive quantum adversaries entangled across instances. Any progress here would also decouple the parameters of Theorem \ref{thm:mainsubexp} from the specific thresholds inherited from \cite{HM25}.

\section{The barrier}
\label{sec:barrier}

\begin{conjecture}[No string-valued characterization of OWSGs]
\label{conj:barrier}
There is a unitary oracle $\mathcal{O}$ relative to which one-way state generators exist, while for every uniform QPT$^{\mathcal{O}}$ sampler over classical strings and every time bound $t$, the problem $\Gap\text{-}\pKq^{t,\mathcal{O}}$ is easy on average for QPT$^{\mathcal{O}}$ adversaries; and more generally every member of a suitably formalized class of string-valued meta-complexity estimation problems is easy. Consequently no relativizing string-based characterization of one-way state generators exists.
\end{conjecture}

\paragraph{Status of the pieces.} The unconditional anchor is the barrier of \cite{CGGH25}: no relativizing characterization of one-way \emph{puzzles} by the average-case hardness of a problem in $\mathsf{NP}$ or $\mathsf{QMA}$; part of Conjecture \ref{conj:barrier} is the definitional task of formalizing ``string-valued meta-complexity problem'' with comparable care one level down the zoo. The construction we envisage is Kretschmer-style \cite{Kre21}: take $\mathcal{O}$ to contain a family of Haar-random unitaries $U_k$, define the state generator $k \mapsto U_k |0\rangle$, and prove security against all QPT$^{\mathcal{O}}$ adversaries at once by concentration of measure over the Haar ensemble combined with hybrid arguments in the Bennett, Bernstein, Brassard, and Vazirani tradition; pair the Haar part with a classical oracle powerful enough, in the spirit of the $\mathsf{P} = \mathsf{NP}$ world of \cite{KQST23}, to estimate the probabilities and complexities of all string-valued samples, which collapses every $\Gap\text{-}\pKq^{t,\mathcal{O}}$ instance by exactly the thresholding of Algorithm \ref{alg:threshold} run in reverse. The delicate point is that unitary security must survive the classical power; that this is possible is precisely the moral of the one-query unitary-synthesis lower bound of \cite{LMW24}, and the reason the proof must be a random-oracle argument rather than a diagonalization: diagonalization defeats countably many machines one at a time, whereas Haar concentration grants the measure-one simultaneous security that a cryptographic statement against all adversaries requires.

\paragraph{What a state-native measure must do.} Circumventing the barrier from below is not hopeless, but it exacts a price the recent EFI characterization already paid \cite{CCCGHJL26}: the complexity measure must take \emph{quantum states} as instances. A density-matrix quantum Kolmogorov complexity fit for purpose must satisfy at least: single-copy estimability, since samplers hand out one copy; a counting or Fano-type bound replacing Proposition \ref{prop:counting}, bounding how many far-apart states admit short descriptions; and meaningful behavior under the unitary invariance that Haar oracles exploit, which string-valued measures lack by construction. The descriptive traditions of \cite{BvDL01,Vit01,Gacs01} supply candidate definitions; which of them supports an analogue of Theorem \ref{thm:codingroot} is, in our view, the right question for the axis below puzzles.

\section{Open problems}
\label{sec:open}

\begin{openproblem}
Prove or refute Conjecture \ref{conj:codingpoly}. A proof completes Theorem \ref{thm:mainpoly} into the exact polynomial-time analogue of the Liu and Pass characterization; a refutation exhibits high-mass quantumly samplable strings that no polynomial-time quantum decoder can reconstruct from optimally short advice, itself a striking object.
\end{openproblem}

\begin{openproblem}
Close the time gap from both ends: improve $t^{*} = 2^{n/2}\poly(n)$ in Theorem \ref{thm:mainsubexp} unconditionally, or prove that black-box use of the sampler cannot do better, making the Grover-optimality heuristic of Remark \ref{rem:obstacles} a theorem.
\end{openproblem}

\begin{openproblem}
Prove Conjecture \ref{conj:amp}, beginning with symmetry of information for $\pKqt$.
\end{openproblem}

\begin{openproblem}
Prove Conjecture \ref{conj:barrier}, and characterize one-way state generators and EFI pairs in the time-bounded regime by state-native measures, extending \cite{CCCGHJL26} down the time axis.
\end{openproblem}

\begin{openproblem}
Quantum Pessiland: does a relativized world exist in which some $\Gap\text{-}\pKqt$ is weakly quantum-average-hard over a quantumly samplable distribution and yet one-way puzzles do not exist? Theorem \ref{thm:mainsubexp} excludes it unrelativized at $t \ge t^{*}$, and Theorem \ref{thm:mainpoly} would exclude it at polynomial time; a relativized construction below $t^{*}$ would delimit the unconditional reach of the present techniques.
\end{openproblem}

\begin{openproblem}
Uniform versus non-uniform: our hardness notions and adversaries are uniform, following \cite{HM25}; determine whether the characterization survives non-uniform adversaries and advice-laden samplers, where the $\log n$ advice of the pseudorandom generator step currently sits.
\end{openproblem}

\begin{openproblem}
Relate the time-bounded compressibility of pseudorandom state ensembles to the known limits on stretching quantum pseudorandomness: coding theorems bound how much a sampler's outputs can be compressed, stretching bounds how much a short seed can be expanded, and a quantitative bridge between the two would tie this program to the structural theory of pseudorandom states.
\end{openproblem}

\section*{Acknowledgments}
The author thanks colleagues at Sunway University for helpful discussions. Any funding sources should be acknowledged here.

\section*{Code availability}
The exact-simulation scripts generating Figures \ref{fig:surprisal}, \ref{fig:decoder}, and \ref{fig:mechanism} (\texttt{make\_figure.py}, \texttt{decoder\_validation.py}, \texttt{mechanism\_validation.py}) are distributed with the manuscript source and require only NumPy and Matplotlib.

\appendix

\section{Numerical illustration of the mass structure}
\label{app:numerics}

The arguments of Sections \ref{sec:proofB} and \ref{sec:proofA} manipulate one object: the distribution of the surprisal $u(x) = \log_2(1/p_{\mathcal{Q}}(x))$ when $x$ is drawn from the sampler. Yes-instances are manufactured below a surprisal threshold, no-instances certified above one, and the coding theorem's regime is the band around $u \approx n$. We exactly simulate a sampler applying $24$ brickwork layers of independent Haar-random two-qubit gates to $|0^{12}\rangle$, computing the full $2^{12}$-dimensional statevector, so the figure contains no sampling noise. Figure \ref{fig:surprisal} shows the surprisal distribution under $\mathcal{Q}$ against the Porter-Thomas prediction for deep random circuits: with $q = 2^n p$, the sampled surprisal density is $\ln(2)\, q^2 e^{-q}$. The measured Shannon entropy is $H(\mathcal{Q}) = 11.380$ bits against the Porter-Thomas value $n - (1 - \gamma)/\ln 2 = 11.390$ bits, $\gamma$ the Euler-Mascheroni constant: an entropy deficit of $0.620$ measured versus $0.610$ predicted. The largest output mass is $2^{-9.06}$, and $55.0\%$ of the mass lies below $H(\mathcal{Q})$. Essentially all mass sits within a few bits of $u = n$: the regime where the classical decoder pays $2^{n - o(n)}$, Theorem \ref{thm:codingroot} pays $2^{n/2 - o(n)}$, and Conjecture \ref{conj:codingpoly} asks for $\poly(n)$. The figure illustrates the objects in play; it is evidence for nothing. The script (\texttt{make\_figure.py}, seed $20260810$) reproduces it in seconds.

\begin{figure}[t]
\centering
\includegraphics[width=0.86\linewidth]{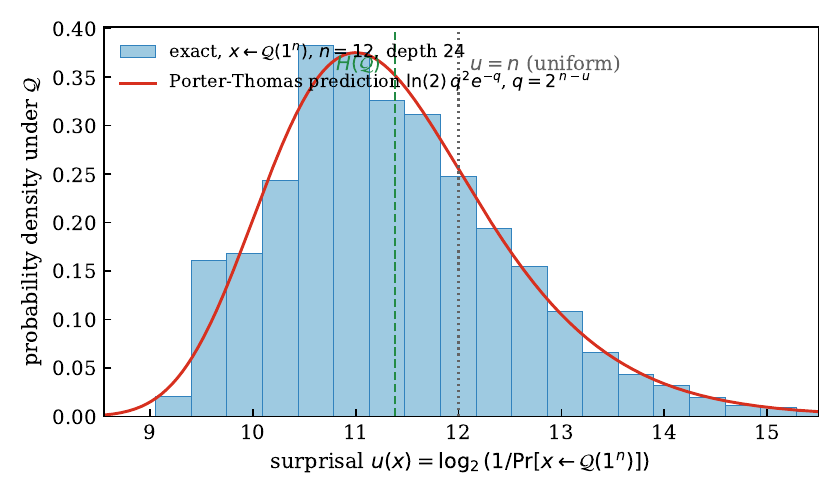}
\caption{Exact surprisal distribution of a depth-$24$ brickwork random-circuit sampler on $n = 12$ qubits (mass-weighted histogram) against the Porter-Thomas prediction (curve). Dashed line: measured entropy $H(\mathcal{Q})$; dotted line: the uniform value $u = n$.}
\label{fig:surprisal}
\end{figure}

\section{Numerical validation of the decoder}
\label{app:decoder}

We implemented the decoder of Theorem \ref{thm:codingroot} exactly as constructed in its proof and tested the four claims the proof rests on. The samplers are exact statevectors on $14$ qubits: a biased product-state family whose output masses span many orders of magnitude, and a depth-$4$ brickwork random circuit. The hash is the polynomial-evaluation family over $\F_{2^m}$ of the proof, with irreducible polynomials found at runtime by Rabin's test; amplitude amplification is implemented through the two reflections directly on the statevector, with the unknown-amplitude schedule of \cite{BBHT98} truncated at $9/\sqrt{\delta}$ rounds. All results below use fixed seed $20260810$ and reproduce in about a minute on a laptop.

First, success: over $26$ targets spanning $\delta \in [2^{-11}, 2^{-2}]$ with $40$ decoder runs each, the overall success rate is $0.915$ and the minimum per-target rate is $0.800$, against the theorem's bound of $2/3$; the residual failures are collisions at the deliberately small hash padding used at this scale, at rate $0.085$, consistent with the collision analysis. Second, scaling (Figure \ref{fig:decoder}, left): the fitted cost exponent in $1/\delta$ is $0.65$ for the quantum decoder against $1.02$ for the classical rejection-sampling baseline on the identical task; the quantum exponent sits above the asymptotic $1/2$ because the schedule constants of \cite{BBHT98} are not negligible at these sizes, while the separation from the classical slope is already an order of magnitude in cost at $\delta = 2^{-11}$. Third, the conditional-distribution step inside the proof's analysis, that amplification preserves relative amplitudes within the marked subspace, holds to machine precision: on the brickwork sampler with $4$ marked strings and $s_a = 1.6 \times 10^{-3}$, the total-variation distance between the post-amplification conditional distribution and $p_{\mathcal{Q}}(\cdot)/s_a$ is $7 \times 10^{-17}$, with amplified marked mass $0.9996$ at the near-optimal round count. Fourth, the collision-mass bound with theorem-matched padding: at $m = \lceil \log(1/\delta) \rceil + 6$ the empirical mean of $B_a$ over $300$ seeds is $3.6 \times 10^{-5}$ against the bound $1.2 \times 10^{-4}$, and the fraction of seeds with $B_a > \delta/20$ is $0.010$, one order below the Markov target of $0.1$.

The experiment validates the construction and its internal lemmas; it is not evidence for or against Conjecture \ref{conj:codingpoly}, whose subject is precisely the regime beyond black-box amplification. The script (\texttt{decoder\_validation.py}) is distributed with the manuscript source.

\begin{figure}[t]
\centering
\includegraphics[width=0.98\linewidth]{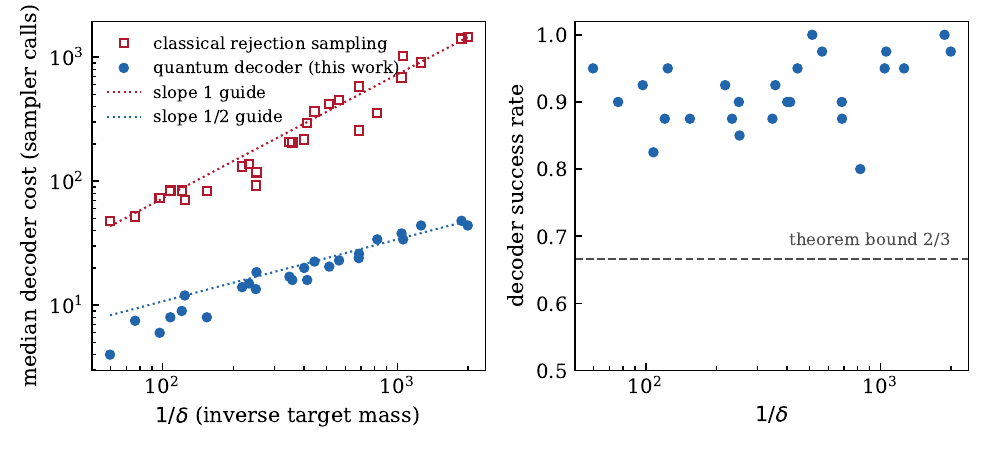}
\caption{Numerical validation of Theorem \ref{thm:codingroot} on exactly simulated $14$-qubit samplers. Left: median decoder cost against $1/\delta$, quantum decoder versus the classical rejection-sampling baseline on the identical task, with slope $1/2$ and slope $1$ guides. Right: per-target success rate against the theorem's $2/3$ bound (dashed).}
\label{fig:decoder}
\end{figure}

\section{Mechanism, tightness, and the partition step}
\label{app:mechanism}

Three further exact experiments probe why the decoder works, why it cannot be accelerated below root time by black-box means, and the one lemma that Theorem \ref{thm:mainsubexp}'s proof imports. The hashing is vectorized through the GF$(2)$-linearity of field multiplication; seed and runtime are as in Appendix \ref{app:decoder}.

\emph{Rotation law and tightness} (Figure \ref{fig:mechanism}, left). For five targets spanning $\delta$ from $2^{-6.6}$ to $2^{-21.1}$, the exact success probability after $j$ amplification rounds equals the closed form $(p_{\mathcal{Q}}(x)/s_a)\sin^2((2j+1)\theta)$ with $\theta = \arcsin\sqrt{s_a}$ to a maximum deviation of $8.5 \times 10^{-13}$ over all targets and rounds; rescaled, all five collapse onto a single sinusoid. The first maximum at $(2j+1)\theta = \pi/2$, that is $j \approx (\pi/4)\sqrt{1/s_a}$, is the round count the decoder uses; since a black-box amplitude-amplification decoder's success is exactly this sinusoid, no such decoder reaches constant success in fewer than $\Theta(\sqrt{1/\delta})$ rounds, which is the empirical face of the optimality discussed in Remark \ref{rem:obstacles} and \cite{ORR13}.

\emph{Scaling decomposition} (Figure \ref{fig:mechanism}, centre). At $n = 16$, over $18$ targets with overall success $0.924$, the median decoder cost fitted against the theorem's parameter $1/\delta$ has exponent $0.575$, already closer to the asymptotic $1/2$ than the $0.65$ seen at $n = 14$ in Appendix \ref{app:decoder}, confirming convergence with system size. Regressed instead against the effective marked mass $1/s_a$, which includes the collision contribution the finite hash cannot avoid, the exponent is $0.556$. The residual gap above $1/2$ in the $1/\delta$ fit is therefore accounted for by hash collisions at the toy padding used here, not by any deviation from amplitude-amplification scaling: measured against its true dynamical predictor, the decoder runs in $\Theta(\sqrt{1/s_a})$ as the analysis requires.

\emph{Mass partition} (Figure \ref{fig:mechanism}, right). Theorem \ref{thm:mainsubexp}'s proof imports the partition step of \cite{HM25} (their Claim 4.12 shape): for a sampler at statistical distance $1 - \epsilon$ from uniform, the mass on strings of probability at least $G \cdot 2^{-n}/\epsilon$ is at least $1 - G - \epsilon$. On the $n = 16$ product sampler, with measured $\epsilon = 0.248$, the bound holds across $G \in [2^{-10}, 2^{-1}]$ with no violation, the measured mass exceeding the bound by a comfortable margin throughout. This is the geometric fact that lets the coding theorem compress generator outputs into yes-instances in the proof of (i) $\Rightarrow$ (ii).

Together with Appendix \ref{app:decoder}, these experiments validate the construction, its scaling, its internal rotation dynamics, and the imported partition lemma, exactly and reproducibly. None bears on Conjecture \ref{conj:codingpoly}. The script (\texttt{mechanism\_validation.py}) accompanies the source.

\begin{figure}[t]
\centering
\includegraphics[width=0.99\linewidth]{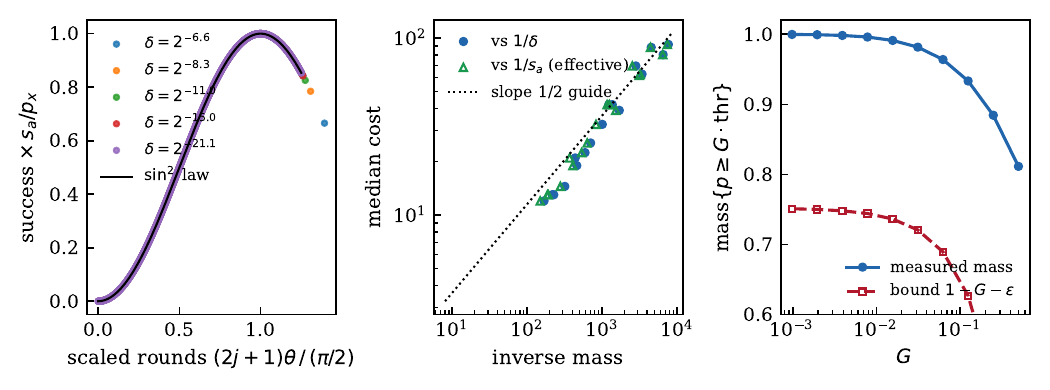}
\caption{Mechanism-level validation. Left: exact success probability after $j$ rounds for five targets spanning fifteen orders of magnitude in mass, collapsed onto the rotation law $\sin^2((2j+1)\theta)$. Centre: median decoder cost at $n=16$ against $1/\delta$ and against the effective marked mass $1/s_a$, with a slope-$1/2$ guide; the $1/s_a$ exponent is $0.556$. Right: measured mass on high-probability strings against the partition bound $1 - G - \epsilon$ of \cite{HM25}.}
\label{fig:mechanism}
\end{figure}


\begin{thebibliography}{99}

\bibitem{LP20} Y.~Liu and R.~Pass. On one-way functions and Kolmogorov complexity. In \emph{Proceedings of the 61st IEEE Symposium on Foundations of Computer Science (FOCS 2020)}, pages 1243-1254, 2020. \href{https://doi.org/10.1109/FOCS46700.2020.00118}{doi:10.1109/FOCS46700.2020.00118}.

\bibitem{IRS21} R.~Ilango, H.~Ren, and R.~Santhanam. Hardness on any samplable distribution suffices: New characterizations of one-way functions by meta-complexity. \emph{Electronic Colloquium on Computational Complexity}, TR21-082, 2021.

\bibitem{GKLO22} H.~Goldberg, V.~Kabanets, Z.~Lu, and I.~C.~Oliveira. Probabilistic Kolmogorov complexity with applications to average-case complexity. In \emph{Proceedings of the 37th Computational Complexity Conference (CCC 2022)}, LIPIcs 234, pages 16:1-16:60, 2022. \href{https://doi.org/10.4230/LIPIcs.CCC.2022.16}{doi:10.4230/LIPIcs.CCC.2022.16}.

\bibitem{LOZ22} Z.~Lu, I.~C.~Oliveira, and M.~Zimand. Optimal coding theorems in time-bounded Kolmogorov complexity. In \emph{Proceedings of the 49th International Colloquium on Automata, Languages, and Programming (ICALP 2022)}, LIPIcs 229, pages 92:1-92:14, 2022. \href{https://doi.org/10.4230/LIPIcs.ICALP.2022.92}{doi:10.4230/LIPIcs.ICALP.2022.92}.

\bibitem{KT24a} D.~Khurana and K.~Tomer. Commitments from quantum one-wayness. In \emph{Proceedings of the 56th ACM Symposium on Theory of Computing (STOC 2024)}, pages 968-978, 2024. \href{https://doi.org/10.1145/3618260.3649654}{doi:10.1145/3618260.3649654}.

\bibitem{KT24b} D.~Khurana and K.~Tomer. Founding quantum cryptography on quantum advantage, or, towards cryptography from \#P-hardness. Cryptology ePrint Archive, Paper 2024/1490, 2024. \href{https://eprint.iacr.org/2024/1490}{eprint.iacr.org/2024/1490}. \href{https://doi.org/10.1145/3717823.3718145}{doi:10.1145/3717823.3718145}.

\bibitem{CGG24} K.-M.~Chung, E.~Goldin, and M.~Gray. On central primitives for quantum cryptography with classical communication. In \emph{Advances in Cryptology, CRYPTO 2024, Part VII}, LNCS 14926, pages 215-248. Springer, 2024.

\bibitem{CGGH25} B.~P.~Cavalar, E.~Goldin, M.~Gray, and P.~Hall. A meta-complexity characterization of quantum cryptography. In \emph{Advances in Cryptology, EUROCRYPT 2025, Part VII}, LNCS 15607, pages 82-107. Springer, 2025. \href{https://arxiv.org/abs/2410.04984}{arXiv:2410.04984}.

\bibitem{HM25} T.~Hiroka and T.~Morimae. Quantum cryptography and meta-complexity. \href{https://arxiv.org/abs/2410.01369}{arXiv:2410.01369}, 2024.

\bibitem{CGGHLP25} B.~Cavalar, E.~Goldin, M.~Gray, P.~Hall, Y.~Liu, and A.~Pelecanos. On the computational hardness of quantum one-wayness. \emph{Quantum}, 9:1679, 2025. \href{https://doi.org/10.22331/q-2025-03-13-1679}{doi:10.22331/q-2025-03-13-1679}.

\bibitem{CCCGHJL26} B.~P.~Cavalar, B.~Chen, A.~Coladangelo, M.~Gray, Z.~Hu, Z.~Ji, and X.~Li. A meta-complexity characterization of minimal quantum cryptography. In \emph{Proceedings of the 58th ACM Symposium on Theory of Computing (STOC 2026)}, pages 675-686, 2026. \href{https://arxiv.org/abs/2510.07859}{arXiv:2510.07859}.

\bibitem{BCQ23} Z.~Brakerski, R.~Canetti, and L.~Qian. On the computational hardness needed for quantum cryptography. In \emph{Proceedings of the 14th Innovations in Theoretical Computer Science Conference (ITCS 2023)}, LIPIcs 251, pages 24:1-24:21, 2023.

\bibitem{JLS18} Z.~Ji, Y.-K.~Liu, and F.~Song. Pseudorandom quantum states. In \emph{Advances in Cryptology, CRYPTO 2018, Part III}, LNCS 10993, pages 126-152. Springer, 2018. \href{https://doi.org/10.1007/978-3-319-96878-0_5}{doi:10.1007/978-3-319-96878-0\_5}.

\bibitem{AQY22} P.~Ananth, L.~Qian, and H.~Yuen. Cryptography from pseudorandom quantum states. In \emph{Advances in Cryptology, CRYPTO 2022, Part I}, LNCS 13507, pages 208-236. Springer, 2022. \href{https://doi.org/10.1007/978-3-031-15802-5_8}{doi:10.1007/978-3-031-15802-5\_8}.

\bibitem{MY22} T.~Morimae and T.~Yamakawa. Quantum commitments and signatures without one-way functions. In \emph{Advances in Cryptology, CRYPTO 2022, Part I}, LNCS 13507, pages 269-295. Springer, 2022. \href{https://doi.org/10.1007/978-3-031-15802-5_10}{doi:10.1007/978-3-031-15802-5\_10}.

\bibitem{MY24} T.~Morimae and T.~Yamakawa. One-wayness in quantum cryptography. In \emph{Proceedings of the 19th Conference on the Theory of Quantum Computation, Communication and Cryptography (TQC 2024)}, LIPIcs 310, pages 4:1-4:21, 2024. \href{https://doi.org/10.4230/LIPIcs.TQC.2024.4}{doi:10.4230/LIPIcs.TQC.2024.4}.

\bibitem{Kre21} W.~Kretschmer. Quantum pseudorandomness and classical complexity. In \emph{Proceedings of the 16th Conference on the Theory of Quantum Computation, Communication and Cryptography (TQC 2021)}, LIPIcs 197, pages 2:1-2:20, 2021. \href{https://doi.org/10.4230/LIPIcs.TQC.2021.2}{doi:10.4230/LIPIcs.TQC.2021.2}.

\bibitem{KQST23} W.~Kretschmer, L.~Qian, M.~Sinha, and A.~Tal. Quantum cryptography in Algorithmica. In \emph{Proceedings of the 55th ACM Symposium on Theory of Computing (STOC 2023)}, pages 1589-1602, 2023. \href{https://doi.org/10.1145/3564246.3585225}{doi:10.1145/3564246.3585225}.

\bibitem{LMW24} A.~Lombardi, F.~Ma, and J.~Wright. A one-query lower bound for unitary synthesis and breaking quantum cryptography. In \emph{Proceedings of the 56th ACM Symposium on Theory of Computing (STOC 2024)}, pages 979-990, 2024. \href{https://doi.org/10.1145/3618260.3649650}{doi:10.1145/3618260.3649650}.

\bibitem{MSY24} T.~Morimae, Y.~Shirakawa, and T.~Yamakawa. Cryptographic characterization of quantum advantage. \href{https://arxiv.org/abs/2410.00499}{arXiv:2410.00499}, 2024.

\bibitem{Yan22} J.~Yan. General properties of quantum bit commitments. In \emph{Advances in Cryptology, ASIACRYPT 2022, Part IV}, LNCS 13794, pages 628-657. Springer, 2022.

\bibitem{BCKM21} J.~Bartusek, A.~Coladangelo, D.~Khurana, and F.~Ma. One-way functions imply secure computation in a quantum world. In \emph{Advances in Cryptology, CRYPTO 2021, Part I}, LNCS 12825, pages 467-496. Springer, 2021.

\bibitem{GLSV21} A.~B.~Grilo, H.~Lin, F.~Song, and V.~Vaikuntanathan. Oblivious transfer is in MiniQCrypt. In \emph{Advances in Cryptology, EUROCRYPT 2021, Part II}, LNCS 12697, pages 531-561. Springer, 2021.

\bibitem{BvDL01} A.~Berthiaume, W.~van Dam, and S.~Laplante. Quantum Kolmogorov complexity. \emph{Journal of Computer and System Sciences}, 63(2):201-221, 2001. \href{https://doi.org/10.1006/jcss.2001.1765}{doi:10.1006/jcss.2001.1765}.

\bibitem{Vit01} P.~M.~B.~Vit\'anyi. Quantum Kolmogorov complexity based on classical descriptions. \emph{IEEE Transactions on Information Theory}, 47(6):2464-2479, 2001. \href{https://doi.org/10.1109/18.945258}{doi:10.1109/18.945258}.

\bibitem{Gacs01} P.~G\'acs. Quantum algorithmic entropy. \emph{Journal of Physics A: Mathematical and General}, 34, 2001. arXiv:quant-ph/0011046. \href{https://doi.org/10.1088/0305-4470/34/35/312}{doi:10.1088/0305-4470/34/35/312}.

\bibitem{LV19} M.~Li and P.~Vit\'anyi. \emph{An Introduction to Kolmogorov Complexity and Its Applications}. Texts in Computer Science, Springer, 4th edition, 2019.

\bibitem{BBHT98} M.~Boyer, G.~Brassard, P.~H{\o}yer, and A.~Tapp. Tight bounds on quantum searching. \emph{Fortschritte der Physik}, 46(4-5):493-505, 1998. arXiv:quant-ph/9605034.

\bibitem{BHMT02} G.~Brassard, P.~H{\o}yer, M.~Mosca, and A.~Tapp. Quantum amplitude amplification and estimation. In \emph{Quantum Computation and Information}, AMS Contemporary Mathematics 305, pages 53-74. American Mathematical Society, 2002. arXiv:quant-ph/0005055.

\bibitem{ORR13} M.~Ozols, M.~Roetteler, and J.~Roland. Quantum rejection sampling. \href{https://arxiv.org/abs/1103.2774}{arXiv:1103.2774}. Conference version in \emph{Proceedings of the 3rd Innovations in Theoretical Computer Science Conference (ITCS 2012)}.

\bibitem{CW79} J.~L.~Carter and M.~N.~Wegman. Universal classes of hash functions. \emph{Journal of Computer and System Sciences}, 18(2):143-154, 1979.

\end{thebibliography}
\end{document}